\documentclass[11pt, one column]{article}

\usepackage[english]{babel}
\usepackage[utf8x]{inputenc}
\usepackage[T1]{fontenc}

\usepackage[margin=1in]{geometry}

\usepackage{graphicx} 
\usepackage[colorlinks=False]{hyperref} 
\usepackage{amsmath}  
\usepackage{amsfonts} 
\usepackage{amssymb}  
\usepackage{amsthm}
\usepackage[shortlabels]{enumitem}
\usepackage{todonotes}
\usepackage{multirow}
\usepackage{float}
\usepackage{pgfplots}
\usepackage{tikz}
\usetikzlibrary{calc}
\usepackage{authblk}
\usepackage{cite}

\newcommand{\C}{\mathcal{C}}

\newcommand{\F}{\mathcal{F}}

\newcommand{\csm}{\textsc{SC-Max}}

\newcommand{\css}{\textsc{SC-Avg}}

\newcommand{\cmix}{\textsc{SC-Mix}}

\newtheorem{theorem}{Theorem}[section]

\newtheorem{lemma}[theorem]{Lemma}

\newtheorem{mechanism}{Mechanism}[section]
\newtheorem{observation}{Observation}[section]
\newtheorem{example}{Example}[section]

\title{Strategyproof Facility Location and Committee Selection\\ with Mixed Max and Sum Agent Types}

\author{Yue Gruszecki}
\author{Elliot Anshelevich}
\affil{Department of Computer Science, Rensselaer Polytechnic Institute}
\affil{lapisgruszecki@gmail.com, eanshel@cs.rpi.edu}

\date{\today}
\begin{document}
\maketitle

\begin{abstract}
    We study strategic facility location, in which $n$ agents are located in an arbitrary metric space, and the goal is to choose $k$ facilities to minimize the total agent cost. The agents can have two types of individual cost functions: max-type where the agent wants to minimize the maximum distance from themselves to any chosen facility, or sum-type where the agent wants to minimize the average distance to the chosen facilities. The agents are self-interested, however, and both the agent location and the agent type may be private information.

    We provide deterministic strategyproof mechanisms for this setting, and prove bounds on their approximation ratio as compared with the solution minimizing the total agent cost. When agent types are private but their locations are known, we prove that an approximation of $\left(3 -\frac{2}{k}\right)$ is always possible, and a better approximation of $\left(\frac{2}{1-k+\sqrt{k^2-k+1}}-1\right)$ is achievable when we know the {\em fraction} of the agents with each type, but not necessarily the type of each individual agent. These bounds hold for arbitrary $k$ and arbitrary metric distances. When agent locations are private, we instead focus on the line metric, and show that a simple generalization of the median mechanism results in an approximation ratio of 3, even for large $k$ and arbitrary mixes of agent types. Our results show the importance of collecting information about agent types vs about their locations, and show that it is possible to produce good outcomes even without such information. 
\end{abstract}

\section{Introduction}
\label{section:intro}
Strategic facility location captures the essence of many types of problems, including those in social choice, clustering, and placing actual physical facilities. Because of this, facility location with self-interested agents has been studied extensively; see, for example, \cite{chan2021mechanism} for a survey. In a classic strategic facility location setting, we are given a metric space containing a set of $n$ agents $\C$ (sometimes called ``clients''), and a set of possible facility locations $\F$. The goal is to select a set of $k$ facilities from $\F$ such that they optimize some objective function. This setting also captures many problems in spatial voting, where voters can be represented as agents and candidates can be represented as facilities \cite{merrill1999unified, enelow1984spatial}.  In the strategic version, agents are self-interested, meaning that they may misreport their information (e.g., their personal cost function, or their true location). 
Then, the goal becomes to form a strategyproof mechanism, meaning that no agent can improve their cost by misreporting their private information, such that this mechanism always chooses facilities which are a good approximation for the optimal solution. 

Much of the existing work on this topic focused on forming strategyproof mechanisms for a line metric (e.g., \cite{li2025strategy, walsh2021strategy, procaccia2013approximate, moulin1980strategy, church2022review}), and for either $k=1$ \cite{CHAN2023114208,chen2020facility, kanellopoulos2023truthful} or $k=2$ \cite{deligkas2025agent,lotfi2024truthful,xu2021two}.  In this paper, we are instead interested in results for arbitrary $k$. The case when $k \geq 2$, so that multiple facilities must be chosen, can be used to represent committee selection, where the goal is to choose a committee $A\subseteq\F$ of size $k$ that minimizes some objective function (e.g., social cost, maximum of all agent costs, etc). The most commonly used objective is the social cost, which is to minimize $\sum_{i\in \C}cost_i(A)$, where $cost_i(A)$ is the individual cost of agent $i$ for the chosen set of facilities $A$. For example, $cost_i(A)$ could be the maximum distance from $i$ to any location in $A$, the average distance from $i$ to locations in $A$, or something else. Previous work has looked at forming good committees when all the agents have the same type of cost function (e.g., all max or all sum). In this work, however, we want to form good outcomes when {\em agents may have different types of cost functions}.

\subsection{Individual Costs of Different Types}
For the simpler case when $k=1$, 
the commonly used social cost objective asks to select a facility $a\in\F$ which minimizes $\sum_{i\in \C}d(i,a)$, with $d(i,a)$ here being the distance from agent $i$'s true location to facility $a$. In other words, the individual cost of agent $i$ (which we denote by $cost_i$) is assumed to be the distance from agent $i$ to the facility, capturing the fact that the agent prefers facilities which are closer to them (or candidates which are more similar to them in the given ideological metric space). 

For our setting with $k\geq 2$, however, there are many different individual cost functions that can be considered. For a set $A\subseteq\F$ of size $k$, how happy would agent $i$ be if this set of facilities were chosen? How should we define $cost_i(A)$?
One commonly considered individual cost for each agent $i$ is the {\em minimum} distance from $i$ to each of the facilities in $A$, i.e, $cost_i(A) = \min_{a\in A}d(i,a)$ \cite{lu2010asymptotically, procaccia2013approximate, fotakis2014power, tang2020mechanism, xu2021two}. This is used, for example, in the $k-$center and $k-$median problem objectives. However, using the minimum distance may not necessarily capture the case where facilities are {\em heterogeneous}, such that agents need to access all of the chosen facilities instead of just one of them. This occurs when agents need to get different resources from different facilities, when facilities are different types (e.g., a park and a post office, and $i$ wants to use both), or when $A$ represents a committee and $i$ wants {\em all} members on the committee to be similar to them, not just a single member. See \cite{serafino2016heterogeneous} for more discussion of such heterogeneous facilities. 
In this paper, we focus on choosing multiple heterogeneous facilities, and thus consider two different types of individual costs. The first cost is defined as $cost_i(A) = \max_{a\in A}d(i,a)$, also known as the {\it max-variant} \cite{chen2020facility, lotfi2024truthful, zhao2023constrained}. The second is defined as $cost_i(A) = \frac{1}{k}\sum_{a\in A}d(i,a)$, also known as the {\it sum-variant} \cite{zhao2024constrained, serafino2016heterogeneous, xu2021two}. To see why these two cost functions make sense for our setting, note that in both variants the agent cares about minimizing their distance to {\em all} the chosen facilities; in the first case it cares about the maximum distance (all facilities should be at most $x$ away from $i$), and in the second case it wants to minimize the average distance.

Agents with max-variant or sum-variant costs have been studied in numerous previous works; most of this existing work only looks at the simpler case of $k=2$, although there are some notable exceptions (see Related Work). More importantly, all of this previous work assumes that all agents have the same type, i.e., they are either all max-type agents, or sum-type agents.\footnote{To the best of our knowledge, the only existing work in strategic facility location which simultaneously allows agents to have different individual costs is \cite{gai2024two}, see Related Work for details.} What if different agents use different cost functions, however? For example, some agents may want to make sure that they are not too far away from the furthest of the facilities, while others want to make sure that the average of the distances to all the facilities is not too far. Moreover, since we are in the strategic setting, we may not know the type of each agent, and desire a strategyproof mechanism which performs well {\em no matter what types the agents are} (as long as each agent is either max or sum type). Our goal in this paper is to form such mechanisms for (1) the case when agent types are private, but we know their exact locations, and (2) the case when agent types are known, but agent locations are private. Our results show the importance of collecting information about agent types vs about their locations, and show that it is possible to produce good outcomes even without such information. 

\subsection{Our Model}
Let $\C$ be a set of agents, and $\F$ be a set of possible facility locations, both contained in some metric space with distance $d$.\footnote{Note that our upper bounds still hold if $\F$ is a multiset, or even if $\F$ is infinite and continuous, but for ease of exposition we will assume that $\F$ is a finite set.} Let $n=|\C|$. Our goal is to choose a set $A$ of $k\geq 2$ facilities from the set $\F$. 

Every agent is either a max-type agent, or a sum-type agent. Let $M\subseteq \C$ be the set of max-type agents, and $S\subseteq \C$ be the set of sum-type agents. For each agent $i$, let their individual cost be $cost_i(A)$ such that for a max-type agent $i\in M$, $$cost_i(A) = \max_{a\in A} d(i,a),$$ and for a sum-type\footnote{Note that the sum-type agents could be more accurately called ``avg-type'' agents, but we call them sum-type to follow the conventions in existing work (e.g., ``sum-variant''). When all agents are the same type, average vs sum does not make a difference in agent costs. In our setting, however, both max and sum agents exist simultaneously, so it makes sense to normalize their costs by taking the average of the distances to the chosen facilities.} agent $i\in S$, $$cost_i(A) = \frac{1}{k} \sum_{a\in A} d(i,a).$$ 

The social cost of a solution $A$ is then simply the sum of the mixed agent types:

$$\cmix(A) = \sum_{i\in \C}cost_i(A) = \sum_{i\in S} \frac{1}{k} \sum_{a\in A} d(i,a) + \sum_{i\in M} \max_{a\in A} d(i,a).$$

It is also convenient to define the social cost values for the case when {\em all} the agents happen to be sum-type or {\em all} happen to be max-type:

\begin{eqnarray*}
\csm(A) = \sum_{i\in \C} \max_{a\in A} d(i,a)\\
\css(A) = \sum_{i\in \C} \frac{1}{k}\sum_{a\in A} d(i,a)
\end{eqnarray*}




In this paper, we are interested in minimizing {\cmix}, while agents can have both their types and locations as private information. Hence, we will form strategyproof mechanisms where one or both of these kinds of information are kept private to the agents. A mechanism $\mathcal{M}$ takes as input the reports of the agents about their private information, and returns a set of $k$ facilities $A\subseteq\F$. A mechanism is {\em strategyproof} if no agent $i$ has incentive to misreport their private information, i.e., no matter what other agents report, the individual cost of agent $i$ will never be improved by lying.

\subsection{Our Contributions}

\begin{figure}[!ht]
    \centering
    \includegraphics[width=3in]{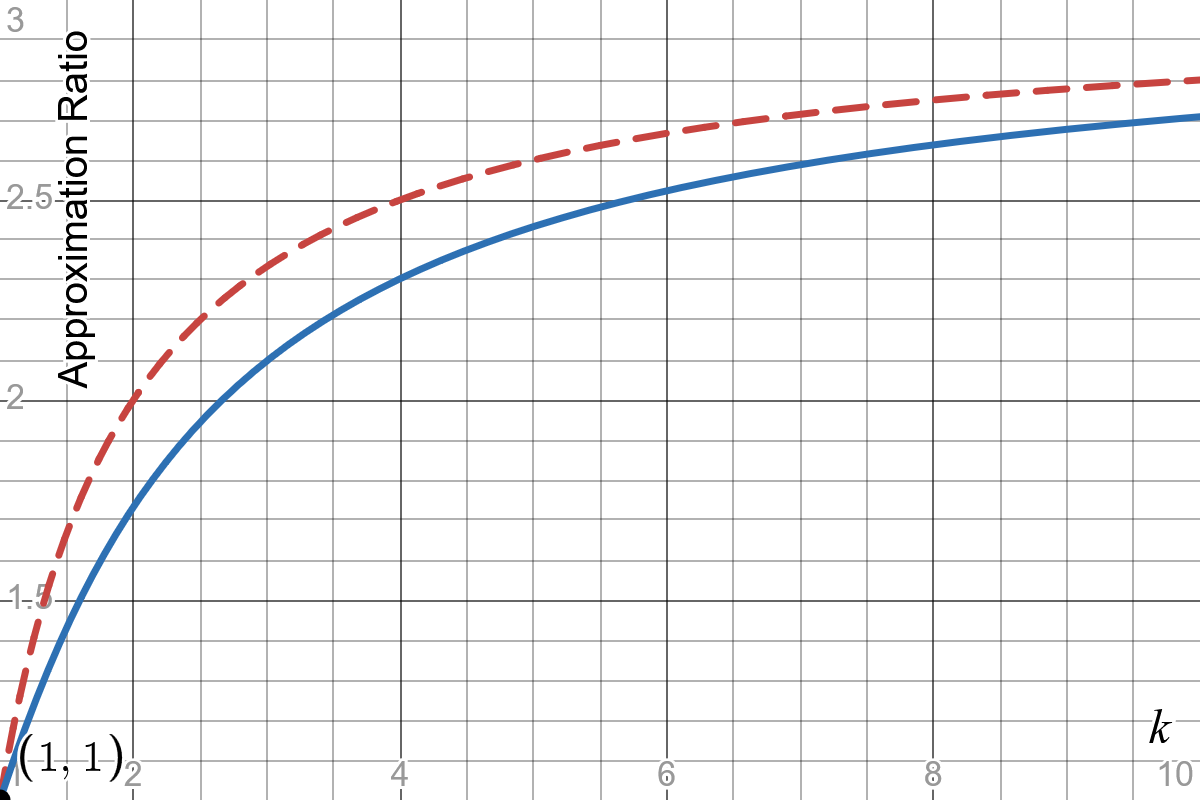}
    \caption{Plots of the approximation ratios for our mechanisms as compared with the optimum {\cmix} solution. The dashed line is the approximation obtained by simply choosing the optimum for {\css}. The solid line is the approximation obtained by choosing either the optimum for {\csm} or {\css}; knowing which one to choose requires knowing the fraction $|M|/n$ of max-type agents, but not which specific agents these are. Both mechanisms do not require access to agent types, and both ratios converge to 3 as $k\rightarrow\infty$. 
    }
    \label{fig:resultk}
\end{figure}

We begin by considering the case when the exact locations of the agents are public knowledge, but the agent types are private and unknown to the mechanism. Note that our results are for an arbitrary metric space (not just for a line metric), and for arbitrary $k$. We give a deterministic strategyproof mechanism that always results in a $\left(3 -\frac{2}{k}\right)$-approximation to the optimum $\cmix$ solution. In fact, this mechanism is extremely simple: it just chooses the optimum solution for \css, which we show is always a good approximation even when many of the agents are actually max-type. See Section \ref{sec:unknownTypes} for a discussion of the techniques used to establish this bound. The above solution can be used when we know absolutely nothing about the agent types (except their being either max or sum agents), but what if we knew the {\em fraction} of max-type agents, without knowing which agents those are? In other words, what if we know $r = |M|/n$ in addition to the agent locations, but not $M$ or $S$ themselves? In that case, we give another strategyproof mechanism that always results in a $\left(\frac{2}{1-k+\sqrt{k^2-k+1}}-1\right)$-approximation in the worst case. Figure \ref{fig:resultk} shows the plots for both bounds w.r.t. $k$. We note that both of these bounds grow from 1 to 3 as $k$ grows larger, but are significantly better than 3 for small $k$. For example, for $k=2$, which is a setting that has received a lot of attention, our mechanisms result in approximations of 2 and $\sqrt{3}$. Even for $k=5$, which is a reasonable size for a set of facilities or a committee, our approximation bounds guarantee solutions within factors of 2.6 and 2.43 of optimum. This also shows that simply knowing the number of each type already improves the approximation ratio significantly as compared with knowing nothing about the agent types. Finally, we provide a lower bound of 1.186 for any deterministic strategyproof mechanism for this setting. While this lower bound is certainly not tight, we remark that it seems to be significantly more difficult to form non-trivial lower bounds for the setting with private agent types and arbitrary $k$ than for the more classic setting with private agent locations.  



Next, we consider the case where the locations of the agents are private. It is already known that even for $k=1$ (and thus for all agent types being the same), it is not possible to form a deterministic strategyproof mechanism with approximation ratio better than $\Omega(n)$ for a general metric space or even for a cycle \cite{schummer2002strategy}. Even for a line metric, it is also not possible to form such a mechanism for $k=1$ with approximation ratio better than 3\cite{feldman2016voting}.
Because these negative results clearly apply to our more general setting as well, we focus on the line metric. We show that for arbitrary $k$, and for an arbitrary mix of agent types between max and sum, there is a simple deterministic strategyproof mechanism which always achieves an approximation ratio of 3 as compared with the optimum for \cmix. This is true even for the case when both agent types and locations are private. In other words, we extend the known 3-approximation bound from \cite{feldman2016voting} for $k=1$ to arbitrarily $k$ and arbitrary mixes of max and sum-type agents. Note that while our mechanism $Med^*$ is a simple generalization of the classic median mechanism for $k=1$, proving that it still results in a 3-approximation for mixed agent types requires new ideas and techniques. This is essentially because the 3-approximation proof from \cite{feldman2016voting} requires ``charging'' costs of some agents to costs of other agents, which is simple to do when all agents have the same type, but becomes significantly more complex when attempting to charge a max agent to a sum agent (since the maximum of $k$ distances can be much larger than the average of those distances). For more details, see Section \ref{sec:agent}. 

Lastly, we consider the case where facilities can be placed anywhere on a line, i.e., the ``unrestricted facilities'' setting. We show that $Med^*$ is strategyproof for any of the cases above on a line metric, and gives the optimal solution for our setting, thus extending the known results for $k=1$ to our more general setting as well.

\subsection{Related Work}
Research on facility location and spatial voting spans many disciplines (including mechanism design) and includes numerous variants, far too extensive to review in full here; for surveys see \cite{farahani2009facility, chan2021mechanism}. The single-facility case ($k=1$) has received the most thorough treatment, while multi-facility settings remain less explored, largely because they introduce significantly greater analytical complexities. With the exception of \cite{gai2024two}, all the work we mention below is on strategyproof mechanisms when agent locations are private, as we are not aware of any other previous work in facility location with private agent costs.\footnote{There is, however, a significant line of work on agents with, often private,  {\em approval} preferences, i.e., where some agents only consider certain facility locations as valid options. See \cite{kanellopoulos2023truthful,serafino2016heterogeneous} and the references therein.} 

For $k=1$, it is well known that, even on a line, the best achievable bound is 3 for the $\css$ objective \cite{feldman2016voting}, and that this bound can be achieved using the median mechanism. The results get much worse beyond a line metric, such that even in circle networks, the approximation ratio for any deterministic strategyproof mechanism is at least $\Omega(n)$, where $n$ is the number of agents \cite{schummer2002strategy}. We show that the approximation ratio of 3 is still achievable even with arbitrary $k$ and mixed agent types, as long as the metric is a line metric. Note that the above results are for ``restricted'' facility locations, where facilities can only be placed at some set $\F$ of locations. If instead facilities can be placed anywhere on a line, then as shown in \cite{procaccia2013approximate},  choosing the median agent location results in the optimal solution for \css, and \cite{schummer2002strategy} gives a strategyproof mechanism that results in the optimal solution for trees. In addition, for $\mathbb{R}^d$ and most $L_p$ metrics, there exist strategyproof mechanisms that can achieve a 3-approximation for this case \cite{gravin2025approximation}. For the very nice setting where facilities can be placed anywhere on a line, we show a simple generalization of the $k=1$ result and prove that a generalization of the median mechanism still forms the optimal solution for $\cmix$, even for mixed agent types.

In this work, we are interested in agents with mixed types of costs, specifically, the maximum distance from an agent to all the facilities (as in \cite{chen2020facility, lotfi2024truthful, zhao2023constrained, han2025compatibility}), and the average (or equivalently, sum, if all agents use the same cost function) of an agent to all the facilities (as in \cite{kanellopoulos2023truthful,lotfi2024truthful, zhao2024constrained,xu2021two, han2025compatibility, serafino2016heterogeneous}). We note that if $k = 1$, then these two costs become the same, which is why this work focuses on the case where multiple facilities should be selected ($k \geq 2$). Most of the work mentioned above only considers the setting with $k=2$, and {\em all} of the work mentioned above only looks at the case when all agents have the same type of cost function (i.e., all sum agents or all max agents). \cite{deligkas2025agent} looks at the case when $k>2$, for the line metric when all agents have the same type (either max or sum). \cite{deligkas2025agent} provides a lower bound of $2-1/k$ and upper bound of 2 for any deterministic strategyproof mechanisms for {\css}, as well as a tight bound of $k+1$ for {\csm}. We note that the bounds in \cite{deligkas2025agent} are incomparable with our 3 bounds, which is due to the following subtle, but important, difference in our settings. In our case, the set of possible facility locations $\F$ is fixed and known, independent of the agents' reported locations. In \cite{deligkas2025agent}, however, the set $\F$ always equals the set of agent locations, and thus changes depending on the agents' reported locations. Having such a fixed set $\F$ allows us to extend the 3-approximation bounds for private agent locations to arbitrary $k$ and arbitrary combinations of max and sum agents, which is provably not possible for the setting studied in \cite{deligkas2025agent}.

All of the above related work is for private agent locations, with all agents having the same type of cost functions. The main results of our paper are instead about the case when agent locations are known, but agent types are private. To our knowledge, strategic facility location with a mix of sum and max type agents has not been considered before. Perhaps the most relevant work to ours would be \cite{gai2024two}, where the authors consider a mix of agents with sum-type agents and {\em min-type} (the cost is the distance from the agent to the {\em closest} facility to them) agents on a line metric. The objective is to minimize the overall social cost, just as in our work. They consider choosing $k=2$ facilities, and assume there is at least one sum type agent. The questions they ask are essentially the same as in our work: what are the approximation bounds achievable for strategyproof mechanisms when either agent types or agent locations are private? When both the agents' locations and types are private, they showed an upper bound of $2n+1$ on this approximation ratio (the lower bound is $n-2$, derived from \cite{fotakis2014power}). When only the agents' locations are private, the upper bound is $n+2$ (the lower bound is also $n-2$). When only the agents' types are private, the upper bound is $2n+1$, and the lower bound is $7/6$. Note that all the upper bounds for the setting in \cite{gai2024two} are $\Omega(n)$, while we show approximation bounds better than 3, even for $k>2$ and an arbitrary metric space instead of a line for the setting when only the types are private. This is entirely due to the fact that we consider max and sum agents, while \cite{gai2024two} consider {\em min} and sum agents: having agents with cost being the distance to the {\em closest} facility results in much worse approximation bounds. Moreover, we believe that max-type agents are much better motivated for heterogeneous facility settings than min-type agents, as described in \cite{chen2020facility}.

\section{Unknown Agent Types}\label{sec:unknownTypes}
We first consider the scenario where the exact locations of the agents are known to the mechanism, but not the agent types. 

Our first result shows that a simple mechanism, which pretends that all agents have the sum type, and chooses the optimum facilities for that case, is a reasonable approximation in terms of the {\cmix} objective. Since this is a strategyproof mechanism (it uses the locations of the agents, which we assume are known in this section, but does not use any information about the true agent types at all), this provides us with a simple way to form a solution which is guaranteed to be a constant approximation to the optimum, for any combination of max and sum agent types. 


\paragraph{Notation:} For compactness of notation, for any agent $i\in \C$ and set of $k$ facilities $A$, we define $d(i,A)=\frac{1}{k}\sum_{a\in A}d(i,a)$ to be the {\em average} of distances from $i$ to $A$. Note that if distances $d$ form a metric, then distances $d(i,A)$ still satisfy the triangle inequality: for agents $i,j\in \C$ we have that $d(i,A)\leq d(i,j)+d(j,A)$ and $d(i,j)\leq d(i,A)+d(j,A)$.  

In this section, let $A$ be the optimal solution for {\css}, $B$ be the optimal solution for {\csm}, and $O$ be the optimal solution for {\cmix}. Our first mechanism ignores any infomation about agent types, and simply returns $A$. To argue that $A$ always has a small {\cmix} cost, we will use the following notation.
For $i \in \C$, let $a_i = \text{argmax}_{a \in A} d(i, a), b_i = \text{argmax}_{b \in B} d(i, b), o_i = \text{argmax}_{o \in O} d(i, o)$. In addition, let $o^* = \operatorname{argmax}_{o \in O} \sum_{i \in \C}d(i, o)$, $a^* = \operatorname{argmax}_{a \in A} \sum_{i \in \C}d(i, a)$, $b^* = \operatorname{argmax}_{b \in B} \sum_{i \in \C}d(i, b)$. 


\vskip 2pt We immediately note the following observation:

\begin{observation}
$\sum_{i\in\C}d(i,a^*)\leq \sum_{i\in\C}d(i,o^*)$    
\label{obs.a*}
\end{observation}

This observation is true since $A$ minimizes $\sum_{i\in \C}\sum_{a\in A} d(i, a) = \sum_{a\in A} \sum_{i\in \C} d(i, a)$, which means it simply consists of the $k$ locations with the smallest average distances to all the agents in $\C$. Thus $\sum_{i\in \C} d(i, a)$ for any $a\in A$ must be at most $\sum_{i\in \C} d(i, a')$ for any $a'\not\in A$. In particular, since $o^*$ has the $k$'th largest distance from the agents of all facilities in $O$, then its distance from the agents must be at least $\sum_{i\in\C}d(i,a^*)$.

Before formally proving a bound on the {\cmix} cost of $A$, we want to motivate our basic approach and proof techniques. Note that if all agents are guaranteed to be of the same type, establishing such bounds is not difficult. If all agents have the sum type, then {\cmix} is the same as {\css}, and so $A$ is indeed the optimum solution. If all agents have the max type, however, it is still not too difficult to obtain such bounds. This was already observed in \cite{han2025compatibility}, but we include a summary of the argument here for completeness. If all agents have the max type, then {\cmix} is the same as {\csm}, so

\[
\begin{aligned}
    \csm(A) &= \sum_{i\in\C}d(i,a_i)\\
    &\leq \frac{1}{n}\sum_{i\in \C}\sum_{j\in\C}[d(i,j) + d(j,a_i)]\\
    &\leq \frac{1}{n}\sum_{i\in \C}\sum_{j\in\C}[d(i,O)+d(j,O)]+\frac{1}{n}\sum_{i\in \C}\sum_{j\in\C}d(j,a^*)\\
    &\leq \frac{1}{n}\sum_{i\in \C}\sum_{j\in\C}[d(i,o_i)+d(j,o_j)]+ \frac{1}{n}\sum_{i\in \C}\sum_{j\in\C}d(j,o^*)\\
    &\leq 2\sum_{i\in \C}d(i,o_i)+ \sum_{i\in \C}d(i,o_i)\\
   & =3\cdot\csm(O)
\end{aligned}
\]
The first inequality above is true due to the triangle inequality, the second due to triangle inequality and definition of $a^*$, the third due to definition of $o_i$ and Observation \ref{obs.a*}, and the last due to definition of $o_i$ again. 

The above provides a relatively simple argument for why $A$ is a good solution to choose even when all the agent types are max instead of sum, since at worst it will result in a 3-approximation. Unfortunately, applying the same argument for the case when agent types are mixed, i.e., can be both sum and max, has several obstacles. One obstacle when attempting to bound $\cmix(A) = \sum_{i\in S} d(i,A) + \sum_{i\in M} d(i,a_i)$, instead of $\csm(A)$, is the term $\sum_{i\in S} d(i,A)$. While we know that $\sum_{i\in \C} d(i,A)\leq \sum_{i\in \C} d(i,O)$ by our choice of $A$, it certainly may be the case that $\sum_{i\in S} d(i,A) > \sum_{i\in S} d(i,O)$. Another obstacle to using the above argument is that if we want to use Observation \ref{obs.a*}, we must look at $\sum_{j\in\C}d(j,a^*)$ over the entire set of agents $\C$, instead of only a subset of agents like $S$ or $M$. Because of this, we necessarily end up with a quantity like $\sum_{i\in \C}d(i,o_i)$, which can be quite large compared to $\cmix(O)= \sum_{i\in S} d(i,O) + \sum_{i\in M} d(i,o_i)$, since for $i\in S$ it can be that $d(i,o_i)$ is much larger that $d(i,O)$. Trying to adapt the above approach directly to mixed agent types results in an approximation bound significantly larger than 3. Because of this, we need a different approach to directly bound $d(i,a_i)$ for agents $i\in M$, which we obtain using the following, somewhat technical, lemma. This lemma allows us to compare the costs of the max-type agents in $A$ with the distances of {\em all} the agents in $A$.


\begin{lemma}
	$\sum_{i \in M} d(i, a_i) \leq \sum_{i\in M}\left[d(i, A) + \frac{1}{n} \sum_{j \in \C} d(j, A) + \frac{k-2}{kn} \sum_{j \in \C} d(j, a_i) \right]$.
	\label{lemma:maxsmaller}
\end{lemma}
\begin{proof}
    For any $a\in A$, let $M(a)$ be the set of agents $i\in M$ so that $a=a_i$, i.e., $a$ is the farthest facility from $i$ in $A$. Break ties arbitrarily, so that the sets $M(a)$ form a partition of $M$. For a fixed $a\in A$, due to the triangle inequality, we have that
    \[
	\begin{aligned}
	\sum_{i \in M(a)} d(i, a) &= \sum_{i \in M(a)} \left[ \frac{1}{k}d(i, a) +  \frac{k-1}{k}d(i, a)\right]\\
		&\leq \sum_{i \in M(a)} \left[ \frac{1}{k}d(i, a) + \sum_{a'\in A, a'\neq a}
        \left(\frac{1}{k}d(i,a') + \frac{1}{k} d(a', a)\right)\right]\\
        &= \sum_{i \in M(a)} d(i,A) + |M(a)|\sum_{a'\in A, a'\neq a}\frac{1}{k} d(a', a)
	\end{aligned}
	\]

    Now, sum the above for all sets $M(a)$. For any $a,a'\in A$, the term $\frac{1}{k}d(a,a')$ appears in this sum exactly $|M(a)|+|M(a')|$ times, so we have that
    \[
	\begin{aligned}
	\sum_{i \in M} d(i, a_i) &= \sum_{a\in A}\sum_{i \in M(a)} d(i, a)\\
    &\leq \sum_{a\in A}\left[ \sum_{i \in M(a)} d(i,A) + |M(a)|\sum_{a'\in A, a'\neq a}\frac{1}{k} d(a', a) \right]\\
    &=\sum_{i\in M}d(i,A) + \frac{1}{k}\sum_{a,a'\in A}(|M(a)|+|M(a')|)d(a',a)
    \end{aligned}
	\]
    Now, if we consider any $j \in \C$, with triangle inequality, we have that 
    \[
    \begin{aligned}
        \frac{1}{k}\sum_{a,a'\in A}(|M(a)|+|M(a')|)d(a',a) &\leq \frac{1}{k}\sum_{a,a'\in A} (|M(a)| + |M(a')|) \left( d(a', j) + d(j, a)\right)\\
        &= \frac{1}{k}\sum_{a\in A}\left[(k-1)|M(a)| + \sum_{a'\in A, a'\neq a}|M(a')|\right]d(j, a)\\
        &= |M|\sum_{a\in A}\frac{1}{k}d(j, a) + \sum_{a \in A}\frac{(k-2)|M(a)|}{k}d(j, a)\\
        &= \sum_{i \in M}d(j,A) + \sum_{a\in A}\sum_{i\in M(a)}\frac{(k-2)}{k}d(j, a)\\
        &= \sum_{i \in M}d(j,A) + \sum_{i \in M}\frac{(k-2)}{k}d(j, a_i)
    \end{aligned}
    \]
    Since the above holds for any $j\in\C$, we have that 
    \[
    \begin{aligned}
        \frac{1}{k}\sum_{a,a'\in A}(|M(a)|+|M(a')|)d(a',a) &\leq \frac{1}{kn}\sum_{j\in \C}\sum_{a,a'\in A} (|M(a)| + |M(a')|) \left( d(a', j) + d(j, a)\right)\\
        &= \frac{1}{n}\sum_{j\in\C}\left(\sum_{i \in M}d(j,A) + \sum_{i \in M}\frac{(k-2)}{k}d(j, a_i)\right)\\
        &= \sum_{i \in M} \left( \sum_{j \in \C}\frac{1}{n}d(j, A) + \sum_{j \in \C}\frac{k-2}{kn}d(j, a_i) \right)
    \end{aligned}
    \]
    Finally, we get 
  
	\[
	\begin{aligned}
		\sum_{i \in M} d(i, a_i) &\leq \sum_{i\in M}d(i,A) + \frac{1}{k}\sum_{a,a'\in A}(|M(a)|+|M(a')|)d(a',a)\\
        &\leq \sum_{i \in M} d(i,A)+ \sum_{i \in M} \left( \sum_{j \in \C}\frac{1}{n}d(j, A) + \sum_{j \in \C}\frac{k-2}{kn}d(j, a_i) \right)\\
        &= \sum_{i\in M}\left[d(i, A) + \frac{1}{n} \sum_{j \in \C} d(j, A) + \frac{k-2}{kn} \sum_{j \in \C} d(j, a_i) \right]
	\end{aligned}
	\]
    as desired.
\end{proof}
Using this lemma, we now proceed to show an upper bound on the approximation ratio resulting from picking solution $A$ as our outcome.

\begin{theorem}
	The optimal solution for {\css} is a $\left(3 -\frac{4}{k} +\frac{|M|}{n}\frac{2}{k}\right)$-approximation for {\cmix}.
	\label{thm:ssbound}
\end{theorem}
\begin{proof}
By Lemma \ref{lemma:maxsmaller}, we have that
	\[
	\begin{aligned}
		\cmix(A) &= \sum_{i\in S} d(i,A) + \sum_{i\in M} d(i,a_i)\\
		&\leq \sum_{i\in S}d(i,A) + \sum_{i\in M}\left[d(i, A) + \frac{1}{n} \sum_{j \in \C} d(j, A) + \frac{k-2}{kn} \sum_{j \in \C} d(j, a_i) \right]\\
		&= \sum_{i \in C}d(i, A)  + \sum_{i\in M}\left[\frac{1}{n} \sum_{j \in \C} d(j, A) + \frac{k-2}{kn} \sum_{j \in \C} d(j, a_i) \right]
	\end{aligned}
	\]
	Recall that $A$ is the optimal solution for $\css$, meaning that $\sum_{i \in C}d(i, A) \leq \sum_{i \in C}d(i, O)$. So we have 
	\[
	\begin{aligned}
		\cmix(A) &\leq \sum_{i \in C}d(i, O)  + \sum_{i\in M}\left[\frac{1}{n} \sum_{j \in \C} d(j, O) + \frac{k-2}{kn} \sum_{j \in \C} d(j, a_i) \right]\\
        &= \sum_{i \in C}d(i, O)  + \frac{|M|}{n}\left[ \sum_{i \in S} d(i, O) + \sum_{i \in M} d(i, O) \right] + \sum_{i\in M}\left[\frac{k-2}{kn} \sum_{j \in \C} d(j, a_i)\right]\\
		&\leq \sum_{i \in C}d(i, O)  + \frac{|M|}{n}\left[ \sum_{i \in S} d(i, O) + \sum_{i \in M} d(i, o_i) \right] + \sum_{i\in M}\left[\frac{k-2}{kn} \sum_{j \in \C} d(j, a_i) \right]\\
		&\leq \left( 1 + \frac{|M|}{n}\right) \sum_{i \in S}d(i, O) + \left( 1 + \frac{|M|}{n}\right) \sum_{i \in M}d(i, o_i) + \sum_{i\in M}\left[\frac{k-2}{kn} \sum_{j \in \C} d(j, a_i) \right]
	\end{aligned}
	\]
	The last two inequalities hold because $o_i = \text{argmax}_{o \in O} d(i, o)$, so $d(i,O) \leq d(i,o_i)$. Now, by Observation \ref{obs.a*}, we have that $\sum_{j \in \C} d(j, a^*) \leq \sum_{j \in \C} d(j, o^*)$, and by definition of $a^*$, we have that $\sum_{j \in \C} d(j, a_i) \leq \sum_{j \in \C} d(j, a^*)$ for any $i\in\C$. Therefore, we can see that 
	\[
	\begin{aligned}
		\sum_{i\in M}\left[\frac{k-2}{kn} \sum_{j \in \C} d(j, a_i) \right] &\leq \sum_{i\in M}\left[\frac{k-2}{kn} \sum_{j \in \C} d(j, a^*) \right]\\
		&\leq \frac{(k-2)|M|}{kn} \sum_{j \in \C} d(j, o^*)\\
		&\leq \frac{(k-2)|M|}{kn} \sum_{j \in \C} d(j, o_j)\\
		&= \frac{(k-2)|M|}{kn} \sum_{j \in S} d(j, o_j) + \frac{(k-2)|M|}{kn} \sum_{j \in M} d(j, o_j)
	\end{aligned}
	\]
	This means that 
	\[
	\begin{aligned}
		\cmix(A) &\leq  \left( 1 + \frac{|M|}{n}\right) \sum_{i \in S} d(i, O) + \left( 1 + \frac{|M|}{n} + \frac{(k-2)|M|}{kn}\right) \sum_{i \in M}d(i, o_i) + \frac{(k-2)|M|}{kn} \sum_{i \in S} d(i, o_i)\\
	\end{aligned}
	\]
    The first two terms above are exactly what we want, since they appear in ${\cmix}(O)$. As for the third term, we can bound it as follows. Using the triangle inequality as usual, we have 
	\[
	\begin{aligned}
		\sum_{i \in S} d(i, o_i) &\leq \sum_{i \in S}d(i, O) + \sum_{i \in S}d(o_i, O) \\
		&\leq \sum_{i \in S}d(i, O) + \sum_{i \in S}\left(\frac{1}{|M|} \sum_{j\in M} d(j, o_i) + \frac{1}{|M|} \sum_{j\in M} d(j, O)\right)\\
		&\leq \sum_{i \in S}d(i, O) + \sum_{i \in S}\left(\frac{1}{|M|} \sum_{j\in M} d(j, o_j) + \frac{1}{|M|} \sum_{j\in M} d(j, o_j)\right)\\
		&= \sum_{i \in S}d(i, O) + \frac{2|S|}{|M|} \sum_{i\in M} d(i, o_i)
	\end{aligned}
	\]
	This means that 
	\[
	\begin{aligned}
		\frac{(k-2)|M|}{kn} \sum_{i \in S} d(i, o_i) &\leq \frac{(k-2)|M|}{kn} \sum_{i \in S}d(i, O) + \left(\frac{2(k-2)|S|}{kn} \right) \sum_{i\in M} d(i, o_i)
	\end{aligned}
	\]
	Finally, we can see that
	\[
	\begin{aligned}
		\cmix(A) &\leq  \left( 1 + \frac{|M|}{n} + \frac{(k-2)|M|}{kn}\right) \sum_{i \in S}d(i, O) + \left( 1 + \frac{|M|}{n} + \left( 1 + \frac{|S|}{n}\right)\frac{k-2}{k}\right) \sum_{i \in M}d(i, o_i)\\
		&\leq \left( 1 + \frac{|M|}{n} + \left( 1 + \frac{|S|}{n}\right)\frac{k-2}{k}\right) \cmix(O)\\
        &=\left(2 -\frac{2}{k} +\left(\frac{|M|}{n}+\frac{|S|}{n}\right)\frac{k-2}{k}+\frac{|M|}{n}\frac{2}{k}\right)\cmix(O)\\
        &=\left(3 -\frac{4}{k} +\frac{|M|}{n}\frac{2}{k}\right)\cmix(O)
        \label{eq.sum_opt_detail_bound}
	\end{aligned}
    \]
	
\end{proof}

Now that we have analyzed the quality of $A$, a natural next step is to consider a very similar strategyproof mechanism such that instead of choosing the optimum solution for {\css}, we choose the optimum solution for {\csm}, $B$. We will then consider how to find a trade-off between the two solutions $A$ and $B$, in order to improve the approximation ratio further. 

\begin{theorem}
	The optimal solution for {\csm} is a $\left( 1 + 2\frac{|S|}{|M|}\right)$-approximation for {\cmix}.
	\label{thm:smbound}
\end{theorem}
\begin{proof}
	We have
	\[
	\begin{aligned}
		\cmix(B) &= \sum_{i\in S} d(i,B) + \sum_{i\in M} d(i,b_i)\\
		&\leq \sum_{i\in S} d(i,b_i) + \sum_{i\in M} d(i,b_i)
	\end{aligned}
	\]
	The last inequality holds by definition of $b_i$. Now, recall that $B$ is the optimal solution for {\csm}, meaning that $\sum_{i\in S} d(i,b_i) + \sum_{i\in M} d(i,b_i) \leq \sum_{i\in S} d(i,o_i) + \sum_{i\in M} d(i,o_i)$. Hence, combined with triangle inequality, we can see that
	\[
	\begin{aligned}
		\cmix(B) &\leq \sum_{i\in S} d(i,o_i) + \sum_{i\in M} d(i,o_i)\\
		&\leq \sum_{i\in S}  d(i,O) + \sum_{i\in S} d(O,o_i) + \sum_{i\in M} d(i,o_i)\\
		&= \cmix(O) + \sum_{i\in S} d(O,o_i)\\
		&\leq \cmix(O) + \sum_{i\in S} \left[\frac{1}{|M|}\sum_{j \in M}d(j,O) + \frac{1}{|M|}\sum_{j \in M}d(j,o_i)\right]\\
		&\leq \cmix(O) + \sum_{i\in S} \left[\frac{1}{|M|}\sum_{j \in M}d(j,o_j) + \frac{1}{|M|}\sum_{j \in M}d(j,o_j)\right]\\
		&= \cmix(O) + \frac{|S|}{|M|}\sum_{j \in M}d(j,o_j) + \frac{|S|}{|M|}\sum_{j \in M}d(j,o_j)\\
		&\leq (1 + 2\frac{|S|}{|M|})\cmix(O)
	\end{aligned}
	\]
\end{proof}

Simply taking the optimum solution for $\csm$ may not result in a good approximation bound for $\cmix$: the upper bound in Theorem \ref{thm:smbound} can be arbitrarily bad as $|M|$ becomes small. In fact, there are simple examples when solution $B$ is $\Omega(n)$ times worse than the optimum solution for $\cmix$. 

Would it be possible to combine the insights from Theorems \ref{thm:ssbound} and \ref{thm:smbound}, however, by sometimes selecting solution $A$ and sometimes solution $B$? If we truly know nothing about the agent types, then there is no way to tell when $B$ should be chosen instead of $A$. But if we know the value $r=|M|/n$, i.e., if we know the {\em fraction} of max players in the system, without knowing the types of any specific players, then better approximation bounds become possible. This is often a reasonable assumption, since estimating the proportion of each type of voter or agent can be done at the macro level, while determining the type of each individual voter can be much more difficult. Specifically, we can show the following theorem.


\begin{theorem}
    There exists a deterministic strategyproof mechanism which always returns a solution within factor $\left(\frac{2}{1-k+\sqrt{k^2-k+1}}-1\right)$ of the optimum for $\cmix$. This mechanism does not require the knowledge of individual agent types, but requires the true value of $r = |M|/n$.
    \label{thm:ratio}
\end{theorem}

\begin{proof}
    Theorems \ref{thm:ssbound} and \ref{thm:smbound} provide us with an approximation bound of at most $3-4/k +2r/k$ if selecting solution $A$, and an approximation bound of $1+2\frac{|S|}{|M|} = 1+2\frac{n-|M|}{|M|} = 1+ 2(\frac{1}{r}-1) = \frac{2}{r}-1$ if selecting solution $B$. Since we know the value of $r$, we can always choose a solution with a bound of $\min\{3-\frac{4}{k} +\frac{2r}{k}, \frac{2}{r}-1\}$ by choosing one of $A$ or $B$ appropriately. By taking the maximum over all $r\in[0,1]$, we see that $\min\{3-\frac{4}{k} +\frac{2r}{k}, \frac{2}{r}-1\}\leq \frac{2}{1-k+\sqrt{k^2-k+1}}-1$, as desired.
\end{proof}

Next, we will show a lower bound of the approximation ratio for {\em any} deterministic strategyproof mechanism. 

\begin{example} Consider an example where $k = 2$ and there are three possible facility locations $A,B,C$ on a line. $A$ and $B$ are on top of (or distance $\epsilon$ away from) each other, and $C$ is distance 2 away from them. All agents are located on top of (or very close to) these locations. Let $a\cdot n$ be the number of max type agents at $A$ or $B$, $b\cdot n$ be the number of max type agents at $C$, $c\cdot n$ be the number of sum type agents at $A$ or $B$, and $d\cdot n$ be the number of sum type agents at $C$, with $p = b+d$, $q = a + c$.
\label{ex.strategyproof}\end{example}

Now consider any deterministic strategyproof mechanism. The locations of all agents are publicly known (see Section \ref{sec:agent} for analysis when these are private), so the agents only report their types. In other words, for Example \ref{ex.strategyproof}, the input to the mechanism can be considered to be just the numbers $a,b,c,d$, since this exactly captures the locations and types of the agents. Note that the values $p$ and $q$ are public knowledge since they can be obtained from the agent locations, but agents can lie about their types and thus misrepresent the values $a,b,c,d$. For Example \ref{ex.strategyproof}, all strategyproof mechanisms must behave as follows: they compare the reported fraction of max type agents at $C$ to a fixed cutoff, and return a solution based only on this. This behavior is formalized in the following lemma. 

\begin{lemma}
    For any deterministic strategyproof mechanism, the must exist some function $g: \mathbb{N}_0 \rightarrow [0,1]\cup\{-\infty\}$ so that for Example \ref{ex.strategyproof}, the mechanism produces the outcome $\{A,B\}$ if $b > g(p)$, and outcome $\{A,C\}$ (or equivalently, $\{B,C\}$) otherwise.
    \label{lemma:sp}
\end{lemma}
\begin{proof}
    Recall that the mechanism knows the exact locations of the agents, but doesn't know their types. This means that we know the values of $p,q$. 
    We also observe that there are only two possible outcomes, $\{A,B\}$ and $\{A,C\}$ (or equivalently, $\{B,C\}$). 

    Let $\mathcal{M}$ be an arbitrary deterministic strategyproof mechanism. This mechanism knows $p$ and $q$, so we can think of it as a function $\mathcal{M}(a,b,c,d)$ that returns either $\{A,B\}$ or $\{A,C\}$. First, we will argue that the outcome of $\mathcal{M}$ cannot depend on $a$ or $c$. Suppose to the contrary, that $\mathcal{M}(a_1,b,c_1,d) = \{A,B\}$ and $\mathcal{M}(a_2,b,c_2,d) = \{A,C\}$ with $p=b+d$ and $q=a_1+c_1=a_2+c_2$. In other words, suppose that when everything else is fixed, different values of $a$ and $c$ change the mechanism outcome. If such a case exists, then there must also exist the same situation where either $a_1n=a_2n+1$ or $a_2n=a_1n+1$, i.e., the two inputs differ by a single player type. Both max-type and sum-type agents located at $A$ or $B$ prefer the solution $\{A,B\}$. Thus, if the number of true types is $a_2$ and $c_2$, each player located at $A$ or $B$ would have incentive to lie about their type to instead misreport and form the input $a_1$ and $c_1$. This is a contradiction with $\mathcal{M}$ being strategyproof, so we know that $\mathcal{M}(a_1,b,c_1,d) = \mathcal{M}(a_2,b,c_2,d)$, and thus we can just consider $\mathcal{M}$ as having $b$ and $d$ as input. Moreover, since we know $p$ and $p=b+d$, then $d=p-b$ and we can just consider $\mathcal{M}$ equivalently to only be a function of $b$.
    
    We now need to show that  $\mathcal{M}$ is a monotone function of $b$. In particular, we want to show that there is some threshold $g(p)$ so that for $b> g(p)$ we have $\mathcal{M}(b)=\{A,B\}$, and below this threshold $\mathcal{M}(b)=\{A,C\}$ (or $\{B,C\}$). Suppose to the contrary, that $\mathcal{M}(b_1) = \{A,B\}$ and $\mathcal{M}(b_2) = \{A,C\}$ with $p=b_1+d_1=b_2+d_2$ and $b_2>b_1$. In other words, suppose that for the same $p$ and $q$ values the mechanism may return $\{A,B\}$ for a smaller $b$ value, but $\{A,C\}$ for a larger $b$ value. If such a case exists, then there must also exist the same situation where $b_2n=b_1n+1$, i.e., the two inputs differ by a single player type.

    Consider any sum-type agent located at $C$ with the true agent type amounts being $b_1$ and $d_1$. We know such an agent exists since $b_1<b_2\leq p$, and so $d_1=p-b_1>0$. 
    This agent prefers the outcome $\{A,C\}$ to the truthful outcome $\mathcal{M}(b_1)=\{A,B\}$. Thus, it has incentive to lie, report that it is a max-type agent, and form the outcome $\mathcal{M}(b_2)=\{A,C\}$ instead. This is a contradiction with $\mathcal{M}$ being strategyproof. Thus there must be some threshold $x$ below which $\mathcal{M}(x)=\{A,C\}$ and above which $\mathcal{M}(x)=\{A,B\}$. Note that $x$ may equal $1$ (in which case $\mathcal{M}$ always returns $\{A,C\}$) or $-\infty$ (in which case $\mathcal{M}$ always returns $\{A,B\}$), so the mechanism may return the same outcome always. As shown above, the value of this threshold cannot depend on $a$ or $c$, but can depend on $p$ and $q=1-p$, thus we can express it as a function $g(p)$.
    \end{proof}

\begin{theorem}
   For every deterministic strategyproof mechanism $\mathcal{M}$, there exists an instance of Example \ref{ex.strategyproof} (i.e., values $a,b,c,d$) so that the outcome of this mechanism 
    has an approximation ratio at least 1.186 as compared with the optimum {\cmix} solution.
    \label{thm:genlowb}
\end{theorem}
\begin{proof}
    Let $\hat{p}=0.593$, and fix an arbitrary deterministic strategyproof mechanism $\mathcal{M}$. By Lemma \ref{lemma:sp}, the behavior of this mechanism on instances of Example \ref{ex.strategyproof} is determined entirely by a single value $g(p)$. Let $x=g(\hat{p})$ for the above value $\hat{p}$.
    
    We will consider three cases: (i) $x=-\infty$, (ii) $x\geq\hat{p}$, and (iii) $0 \leq x < \hat{p}$. First, we will consider case (i) $x=-\infty$. By Lemma \ref{lemma:sp}, the outcome of $\mathcal{M}$ would be $\{A,B\}$ always. 
    For the instance of Example \ref{ex.strategyproof} with $a=0$, $b=0$, $c=1-\hat{p}$, $d=\hat{p}$, the approximation ratio of $\mathcal{M}$ compared to the optimum solution would be 
    \[
    \begin{aligned}
        \frac{\cmix(\{A,B\})}{\cmix(\{A,C\})} &= \frac{0an+0cn+2bn+2dn}{2an+cn+2bn+dn}\\
        &= \frac{2\hat{p}}{(1-\hat{p})+\hat{p}}\\
        &= 2\hat{p} = 1.186.
    \end{aligned}
    \]
    Thus, for mechanisms with $g(\hat{p})=-\infty$, there are instances where the approximation ratio is at least $1.186$.

    Now, we will consider case (ii) $x\geq \hat{p}$. 
    Consider the instance of Example \ref{ex.strategyproof} with $a=1-\hat{p}$, $b=\hat{p}$, $c=0$, $d=0$. By Lemma \ref{lemma:sp}, since $b\leq x$, the outcome $\mathcal{M}(b)$ would be $\{A,C\}$. The approximation ratio of $\mathcal{M}$ for this instance would be 
    \[
    \begin{aligned}
        \frac{\cmix(\{A,C\})}{\cmix(\{A,B\})} &= \frac{2an+cn+2bn + dn}{2bn+2dn}\\
        &= \frac{2(1-\hat{p})+2\hat{p}}{2\hat{p}}\\
        &= \frac{1}{\hat{p}}\approx 1.686 > 1.186.
    \end{aligned}
    \]
    Thus, for mechanisms with $g(\hat{p})\geq \hat{p}$, there are instances where the approximation ratio is larger than $1.186$.

    Finally, we consider case (iii) $0 \leq x < \hat{p}$. 
    Consider the instance of Example \ref{ex.strategyproof} with $a=1-\hat{p}$, $b=x$, $c=0$, $d=\hat{p}-x$. Then since $b\leq x$, the outcome $\mathcal{M}(b)$ would be $\{A,C\}$. 
    The approximation ratio of $\mathcal{M}$ for this instance would be
    \[
    \begin{aligned}
        \frac{\cmix(\{A,C\})}{\cmix(\{A,B\})} &= \frac{2an+cn+2bn + dn}{2bn+2dn}\\
        &= \frac{2(1-\hat{p})+2x+(\hat{p}-x)}{2\hat{p}}\\
        &= \frac{2-\hat{p}+x}{2\hat{p}}\\
        &\geq \frac{2-\hat{p}}{2\hat{p}}> 1.186.
    \end{aligned}
    \]

Thus for all possible values of $g(\hat{p})$ (and thus for all deterministic strategyproof mechanisms due to Lemma \ref{lemma:sp}), there exist instances of Example \ref{ex.strategyproof} in which the mechanism returns a solution with cost that is at least 1.186 times that of optimum.
\end{proof}

\section{Unknown Agent Locations}\label{sec:agent}
In this section, we will consider the case where we do not know the agents' true locations. For $k=1$, this case reduces to the classic strategic facility location setting, where the goal is to place a single facility $a\in \F$ to minimize $\sum_{i\in \C}d(i,a)$ without knowing the true agent locations. For a general metric space, it is well known that the approximation ratio of any deterministic strategyproof mechanism is at least $\Omega(n)$ \cite{schummer2002strategy}. Even for a simple 1D metric where all agents are located on a line, all deterministic strategyproof mechanisms have an approximation ratio of at least 3.\footnote{Note that these lower bounds are for the setting with ``restricted'' facility locations that must be chosen from a set $\F$, which is the setting in the majority of our paper. Better bounds are known for the case where facilities can be placed anywhere on the line, which is the setting we investigate in Section \ref{sec:unres}.} Since the setting we study is a generalization of this, we immediately know that no better approximations are possible, even if all the agent types are public knowledge. It is not clear, however, if the upper bound of 3 still holds for the case when $k>1$ and agents have mixed types, since this is a much more general setting than for $k=1$. In this section, we focus on the setting where the metric space is assumed to be a line metric, as in much of the existing work on this topic. We show how to extend the known results for $k=1$ to the setting where $k$ can be arbitrarily large, and agents may have different types, which may be private. Doing this requires some new ideas and techniques, essentially because the 3-approximation proof from \cite{feldman2016voting} requires ``charging'' costs of some agents to costs of other agents. This is relatively simple to do when all agents have the same type, but as discussed in the proof of Lemma \ref{lem:charging}, becomes significantly more complex when attempting to charge a max agent to a sum agent
(since the maximum of $k$ distances can be much larger than the average of those distances). 

For an agent $i$, let $i$ be their true location and $x_i$ be their reported location. In the rest of this section we assume that all agents and facilities are located on a line and denote an outcome by $Y = \{y^1, y^2, \cdots, y^k\}$ where $y^1 < y^2 < \cdots< y^k$. We begin by showing a useful fact about optimal solutions of $\cmix$. Define a {\em contiguous solution} $Y=\{y^1, y^2, \cdots, y^k\}$ to be any set of $k$ facilities so that they are contiguous locations in $\F$, i.e., there is no location $z\in F$ such that $y^1 < z < y^k$ and $z\not\in Y$. 

\begin{theorem}
    On a line, there must exist an optimal solution for $\cmix$ such that all chosen locations are contiguous. 
    \label{thm:optnextto}
\end{theorem}
To show this, we will first show a useful lemma. 
\begin{lemma} \label{lemma:move}
    Consider a solution $Y = \{y^1, y^2, \cdots, y^k\}$ such that there exists some $z \notin Y, z\in \F, y^1 < z < y^k$. Then, either $Y_k = \{y^2, \cdots, z, \cdots y^k\}$ or $Y_1 = \{y^1, \cdots, z, \cdots y^{k-1}\}$ would result in having a $\cmix$ cost no larger than $\cmix(Y)$.
\end{lemma}
\begin{proof}
   Let the set of sum type agents be $S$, and the set of max type agents be $M$. We will first consider an arbitrary $i\in M$. Since it is a max-type agent, and all facilities are on a line, then we know that $cost_i(Y)$ is either $d(i,y^1)$ or $d(i,y^k)$, i.e., $cost_i(Y)=\max\{d(i,y^1),d(i,y^k)\}$. Since $y^2$ and $y^{k-1}$ are located in between $y^1$ and $y^k$, we also know that $d(i,y^2)\leq\max\{d(i,y^1),d(i,y^k)\}$ and $d(i,y^{k-1})\leq\max\{d(i,y^1),d(i,y^k)\}$. Thus, we have that $cost_i(Y_1)=\max\{d(i,y^1),d(i,y^{k-1})\}\leq \max\{d(i,y^1),d(i,y^k)\} =cost_i(Y)$, and similarly for $cost_i(Y_k)$. 
   
   Therefore, we can conclude that 
   $$\sum_{i\in M} cost_i(Y_1) \leq \sum_{i\in M} cost_i(Y)$$
   $$\sum_{i\in M} cost_i(Y_k) \leq \sum_{i\in M} cost_i(Y)$$

   Next, we will consider $i\in S$. Let the set of agents $i \in S$ such that $i \leq z $ be $L$ and the set of agents $i \in S$ such that $i > z $ be $R$. We will consider two cases: $|L| \geq |R|$ and $|L| \leq |R|$. First, we consider the case where $|L| \geq |R|$. We note that 
   \[
   \begin{aligned}
       \sum_{i \in L} cost_i(Y_k) - \sum_{i \in L} cost_i(Y) &= \frac{1}{k}\sum_{i \in L} \left(d(i, z)-d(i,y^k)\right)\\
       &= -\frac{|L|}{k}d(z,y^k)\\
       \sum_{i \in R} cost_i(Y_k) - \sum_{i \in R} cost_i(Y) &= \frac{1}{k}\sum_{i \in R} \left(d(i, z )-d(i,y^k)\right)\\
       &\leq \frac{|R|}{k}d(z,y^k)\\
   \end{aligned}
   \]
   Recall that $|L| \geq |R|$, so we have 
   \[
   \begin{aligned}
       \sum_{i \in S} cost_i(Y_k) - \sum_{i \in S} cost_i(Y) &= \sum_{i \in L} cost_i(Y_k) - \sum_{i \in L} cost_i(Y) + \sum_{i \in R} cost_i(Y_k) - \sum_{i \in R} cost_i(Y)\\
       &\leq -|L|d(z,y^k) + |R|d(z,y^k)\\
       &\leq -|L|d(z,y^k) + |L|d(z,y^k)\\
       &= 0
   \end{aligned}
   \]
   This means that $\sum_{i\in S} cost_i(Y_k) \leq \sum_{i\in S} cost_i(Y)$. Combined with $\sum_{i\in M} cost_i(Y_k) \leq \sum_{i\in M} cost_i(Y)$, this means that $\cmix(Y_k)\leq \cmix(Y)$. Now, for the second case, $|L| \leq |R|$, using a symmetric argument, we have that $\cmix(Y_1)\leq \cmix(Y)$. Therefore, we can conclude that either $Y_1$ or $Y_k$ would result in having a $\cmix$ cost no larger than $\cmix(Y)$.
\end{proof}

Using this lemma, we can now prove Theorem \ref{thm:optnextto}.

\begin{proof}[Proof for Theorem \ref{thm:optnextto}]
    Let an optimal solution be $O = \{o^1, o^2, \cdots, o^k\}$. If all the $k$ locations are contiguous, then we are done. Therefore, we assume otherwise: there must exist some $z\in \F$ such that $z \notin O, o^1 < z < o^k$. Then, by Lemma \ref{lemma:move}, we can either move $o^1$ or $o^k$ to $z$ so that the resulting solution $O'$ would either (1) have $\cmix(O') < \cmix(O)$, a contradiction, or (2) $\cmix(O') = \cmix(O)$. In the second case, we can then repeat the same procedure on $O'$ until the resulting solution is contiguous. 
\end{proof}

Now we are going to introduce a strategyproof mechanism for the case where we do not know the true locations of the clients. This mechanism is a simple generalization of the median mechanism \cite{freeman2021truthful, moulin1980strategy, BARBERA2011731, procaccia2013approximate} for $k=1$, which chooses the facility closest to the median agent, and is known to be strategyproof. It has also been shown to have an approximation factor of at most 3 for 1D metrics (see \cite{feldman2016voting}). It is not difficult to see that our mechanism is still strategyproof for $k>1$ as well, but we include the argument below for completeness (Theorem \ref{thm:medsp}).  Showing that this mechanism still has an approximation factor of at most 3, however, and that this holds even for agents of mixed types, is not so easy, and requires new techniques as compared with the $k=1$ case. Most of this section is devoted to proving this (Theorem \ref{thm:types3}).

\begin{mechanism}
	Given the reported agent locations $(\vec{x}_i)$ on a line, mechanism $Med^*$ returns $k$ closest  facilities from the median agent location. More formally, we have:
    \[
    \begin{aligned}
        Med^*(\vec{x}) = \text{argmin}_{A\subseteq \F, |A| = k}\sum_{a\in A} d(a, \text{median}(\vec{x}))
    \end{aligned}
    \]
    Note that when the number of agents is even, $\text{median}(\vec{x})$ would take the average of its two medians. 
\label{mech:med}
\end{mechanism}

Before proceeding to the main results for this section, we first prove the following useful lemmas about agent costs in contiguous solutions.

\begin{lemma}
    Let $Y$ and $Z$ be arbitrary contiguous solutions of size $k$, with $|Y\cap Z|=k-p$ for some $p\leq k$. Let $A=Y\setminus Z=\{a^1,a^2,\ldots, a^p\}$, $B=Z\setminus Y=\{b^1,b^2,\ldots, b^p\}$, and $C=Y\cap Z = \{c^1,c^2,\ldots, c^{k-p}\}$. Then, it must be either that $$a^1< a^2<\ldots< a^p< c^1< c^2<\ldots< c^{k-p}< b^1< b^2<\ldots< b^p$$ or
    $$b^1< b^2<\ldots< b^p< c^1< c^2<\ldots< c^{k-p}< a^1< a^2<\ldots< a^p.$$ \label{lem:1contiguous}
\end{lemma}

\begin{proof}
    Without loss of generality, suppose that $a^1< b^1$. First, note that $a^p < b^p$. If that was not the case, then $a^1<b^1<b^p<a^p$, which is a contradiction to $Y$ being contiguous, since $b^1$ and $b^p$ are not in $Y$, but are located between two locations in $Y$. Thus, we have that $a^1 < b^1$ and $a^p < b^p$. 

    If $c^1< a^p$, that would mean that solution $Z$ has a location $a^p\not\in Z$ which is strictly between the first and last locations of $Z$, i.e., $c^1 < a^p < b^p$. Thus, we also know that $a^p < c^1$. Similarly, we have that $c^{k-p}<b^1$, since otherwise $b^1$ being between $a^1$ and $c^{k-p}$ contradicts the contiguity of $Y$. Thus, we have shown that $a^1 < a^p < c^1 < c^{k-p} < b^1 < b^p$, as desired. In the special case when $|C|=0$, we also see that $a^p < b^1$ by the same argument.
    %
\end{proof}

\begin{lemma}
For some reported locations $\vec{x}$, consider the contiguous solution $Med^*(\vec{x})$ and another arbitrary contiguous solution $Z\neq Med^*(\vec{x})$. Let $|Med^*(\vec{x})\cap Z|=k-p$ for any $p \leq k$, and let $A=Med^*(\vec{x})\setminus Z = \{a^1,a^2,\ldots, a^p\}$ and $B=Z\setminus Med^*(\vec{x}) = \{b^1,b^2,\ldots, b^p\}$. Without loss of generality, suppose that $a^1< a^2<\ldots< a^p< b^1< b^2<\ldots< b^p$. Then for any agent $i$ located to the left of $median(\vec{x})$ (i.e., $i\leq median(\vec{x})$), it must hold that $d(i,a^j)\leq d(i,b^j)$ for all $j=1..p$.
    \label{lem:2contiguous}
\end{lemma}

\begin{proof}
Since both of the solutions are contiguous, and we are not considering their intersection, then by Lemma \ref{lem:1contiguous} it must be that $a^1< a^2<\ldots a^p< b^1< b^2< \ldots< b^p$. Consider $d(i,a^j)$ for any $j\leq p$. We compare the distance $d(i,a^j)$ with the distance $d(i,b^j)$. If $i$ is to the left of $a^j$, then $i\leq a^j < b^{j}$, so clearly $d(i,a^j) < d(i,b^j)$. If instead $a^j\leq i$, then $a^j\leq i\leq median(\vec{x})\leq b^{j}$. To see why $median(\vec{x})\leq b^{j}$, note that $a^j$ is included in the closest $k$ locations to $median(\vec{x})$, and $b^{j}$ is not included in these locations, so $median(\vec{x})$ cannot be to the right of $b^{j}$ and thus farther from $a^j$ than  from $b^{j}$. Since $a^j\leq i\leq median(\vec{x})\leq b^{j}$, then we have that $d(i,a^j)\leq d(median(\vec{x}),a^j)\leq d(median(\vec{x}),b^{j})\leq d(i,b^{j})$. The reason why $d(median(\vec{x}),a^j)\leq d(median(\vec{x}),b^{j})$ is once again because $a^j$ has to be closer to $median(\vec{x})$ than $b^{j}$ to be included in the set $Med^*(\vec{x})$ instead of $b^{j}$. Thus we have shown that for any $j\leq p$, we have $d(i,a^j)\leq d(i,b^{j})$. 
\end{proof}

\begin{lemma}
Consider the same setting as in Lemma \ref{lem:2contiguous}. Then, for any agent $i$ located to the left of $median(\vec{x})$ (i.e., $i\leq median(\vec{x})$), it must hold that $cost_i(Med^*(\vec{x}))\leq cost_i(Z)$. This holds for both max and sum type agents.
    \label{lem:3contiguous}
\end{lemma}

\begin{proof}
Let $Y=Med^*(\vec{x})$. Let $i$ be a max agent, and suppose $cost_i(Y)=d(i,y)$ for some $y\in Y$. If $y\in Y\cap Z$, then $y$ is contained in $Z$ as well, so $cost_i(Z)=\max_{z\in Z}d(i,z)\geq d(i,y)$, so $cost_i(Y)\leq cost_i(Z)$. If instead $y\in A=Y\setminus Z$, it must equal some $a^j$, and by Lemma \ref{lem:2contiguous} we have that $d(i,a^j)\leq d(i,b^j)\leq cost_i(Z)$, since $b^j\in B\subseteq Z$. So $cost_i(Y)\leq cost_i(Z)$. 

Similarly, let $i$ be a sum agent, so $cost_i(Y)=d(i,Y)$. By Lemma \ref{lem:2contiguous}, we know that for each $a^j\in A$, we have $d(i,a^j)\leq d(i,b^{j})$. Thus in total, we have that 
$$\sum_{y\in Y} d(i,y)=  \sum_{j=1}^{p}d(i,a^j) +\sum_{y\in Z\cap Y} d(i,y)\leq \sum_{j=1}^p d(i,b^j) +\sum_{y\in Z\cap Y} d(i,y) = \sum_{z\in Z} d(i,z),$$
and thus $cost_i(Y)\leq cost_i(Z)$ for all sum agents as well. 
\end{proof}

Using the above lemmas which establish the main properties of contiguous solutions, it is now simple to prove the following theorem. 

\begin{theorem}\label{thm:medsp}
    $Med^*$ is strategyproof. 
\end{theorem}
\begin{proof}
Suppose that the reported locations are $\vec{x}$, and consider any agent $i$. Suppose $x_i=i$, i.e., $i$ reports its true location. Without loss of generality, assume that $i$ is to the left of (or equal to) $median(\vec{x})$. 

If $i$ misreports its position to be any $x_i'$ which is at most $median(\vec{x})$, then the result of the mechanism does not change, i.e., $Med^*(\vec{x}_{-i}, x_i')= Med^*(\vec{x})$, since $median(\vec{x}_{-i}, x_i')=median(\vec{x})$. Thus, $i$ does not improve its cost by reporting $x_i'$ instead of reporting truthfully. 

Now suppose that instead $i$ misreports its position to be $x_i'$ to the right of $median(\vec{x})$, so that the median agent location actually changes. Let $Y=Med^*(\vec{x})$ and $Z=Med^*(\vec{x}_{-i}, x_i')$.  If $Z=Y$ then once again $i$ does not improve its cost by reporting $x_i'$, so we can assume $Z\neq Y$ and $median(\vec{x})<median(\vec{x}_{-i}, x_i')$. We now apply Lemma \ref{lem:3contiguous} to $Med^*(\vec{x})=Y$ and $Z$. The lemma tells us that, 
for  both max and sum agents, the cost of agent $i$ obtained by mis-reporting is no better than the cost it obtains by telling the truth, i.e., $cost_i(Y)\leq cost_i(Z)$.

Since both max agents and sum agents have no incentive to misreport their locations, we conclude that $Med^*$ is strategyproof.
\end{proof}

Next, we will prove an upper bound on the approximation ratio of $Med^*$.
\begin{theorem}
    When agents report truthfully, $Med^*$ always selects an outcome which is at most a 3-approximation for \cmix. 
    \label{thm:types3}
\end{theorem}
\begin{proof}
    Let the reported locations of the agents be $\vec{x}$. By Theorem \ref{thm:medsp}, $Med^*$ is strategyproof, meaning that we can assume $\vec{x}$ are the true locations of the agents. Now, let $Y = Med^*(\vec{x})$, and $O$ be an optimal solution such that all the $k$ facility locations in $O$ are contiguous. Note that such an optimal solution must exist by Theorem \ref{thm:optnextto}. Assume $Y \neq O$, otherwise we are done. 
    
    Suppose that $|Y\cap O|=k-p$ for some $p\leq k$. Note that $p\geq 1$ since $Y\neq O$. Let $A=Y\setminus O=\{a^1,a^2,\ldots, a^p\}$, $B=O\setminus Y=\{b^1,b^2,\ldots, b^p\}$, and $C=Y\cap O = \{c^1,c^2,\ldots, c^{k-p}\}$. Without loss of generality, assume that $a^1<b^1$, then by Lemma \ref{lem:1contiguous} we know that 
    $$a^1< a^2<\ldots< a^p< c^1< c^2<\ldots< c^{k-p}< b^1< b^2<\ldots< b^p.$$

    Let $L$ be the set of agents $i$ such that $i\leq median(\vec{x})$, and $R$ be the set of agents to the right of the median. By definition of median, we know that $|L|\geq |R|$. We now want to prove that $\sum_{i\in\C}cost_i(Y)\leq 3\sum_{i\in \C}cost_i(O)$. 

    For any agent $i\in L$, Lemma \ref{lem:3contiguous} tells us exactly that $cost_i(Y)\leq cost_i(O)$. This may not be true, however, for agents in $R$. We will need to charge some of the costs of agents in $R$ to agents in $L$; to do this we define an arbitrary one-to-one mapping $m:R\rightarrow L$. We know this mapping can be one-to-one since $|L|\geq |R|$. 

    \begin{lemma}
        For any agent $i\in R$, we have that 
        $$cost_i(Y)+cost_{m(i)}(Y)\leq 3[cost_i(O)+cost_{m(i)}(O)]$$ \label{lem:charging}
    \end{lemma}

    \begin{proof}
        This lemma holds no matter what the types of $i$ and $m(i)$ are. We prove this for each of the four possible cases separately. Note that the treatment of Case 4 is significantly different from Cases 1-3.

        \noindent{\bf Case 1: } $i$ and $m(i)$ are both sum-type agents. Then, using the triangle inequality and the fact that for $m(i)\in L$ we already know $cost_{m(i)}(Y)\leq cost_{m(i)}(O)$ due to Lemma \ref{lem:3contiguous}, we have that
        \[
        \begin{aligned}
            cost_i(Y)+cost_{m(i)}(Y) &= d(i,Y) + d(m(i),Y)\\
            &\leq d(i,O) + d(m(i),O) + 2d(m(i),Y)\\
            &\leq d(i,O) + 3d(m(i),O)\\
            &= cost_i(O) + 3cost_{m(i)}(O)\\
            &\leq 3[cost_i(O)+cost_{m(i)}(O)]
        \end{aligned}
        \]

         \noindent{\bf Case 2: } $i$ and $m(i)$ are both max-type agents. As in Section \ref{sec:unknownTypes}, we denote by $y_i$ the farthest facility in a solution $Y$ from $i$. Then, similarly to Case 1, we have
        \[
        \begin{aligned}
            cost_i(Y)+cost_{m(i)}(Y) &= d(i,y_i) + d(m(i),y_{m(i)})\\
            &\leq d(i,O) + d(m(i),O) + d(m(i),y_i) + d(m(i),y_{m(i)})\\
            &\leq d(i,o_i) + d(m(i),o_{m(i)}) + d(m(i),y_{m(i)}) + d(m(i),y_{m(i)})\\
            &\leq cost_i(O) + 3cost_{m(i)}(O)\\
            &\leq 3[cost_i(O)+cost_{m(i)}(O)]
        \end{aligned}
        \]
        The first inequality above is due to the triangle inequality, the second due to the fact that average is at most the maximum, and the last is due to the fact that for $m(i)\in L$, we know that $cost_{m(i)}(Y)\leq cost_{m(i)}(O)$.

        \noindent{\bf Case 3: } $i$ is a sum-type agent and $m(i)$ is a max-type agent. Then, we have that
        \[
        \begin{aligned}
            cost_i(Y)+cost_{m(i)}(Y) &= d(i,Y) + d(m(i),y_{m(i)})\\
            &\leq d(i,O) + d(m(i),O) + d(m(i),Y) + d(m(i),y_{m(i)})\\
            &\leq d(i,O) + d(m(i),o_{m(i)}) + d(m(i),y_{m(i)}) + d(m(i),y_{m(i)})\\
            &\leq cost_i(O) + 3cost_{m(i)}(O)\\
            &\leq 3[cost_i(O)+cost_{m(i)}(O)]
        \end{aligned}
        \]
        The first inequality above is due to the triangle inequality, the second due to the fact that average is at most the maximum, and the last is due to the fact that for $m(i)\in L$, we know that $cost_{m(i)}(Y)\leq cost_{m(i)}(O)$.

        \noindent{\bf Case 4: } $i$ is a max-type agent and $m(i)$ is a sum-type agent. In the previous three cases, we did not really need to use the fact that agents and facilities are located on a line. This case, however, requires a different analysis. The reason why this case is different is because we can always charge sum to max (the average distance from an agent to facilities is at most the maximum distance), but cannot do the opposite. 
        
        Recall our definitions of sets $A$, $B$, and $C$, and that 
        $$a^1< a^2<\ldots< a^p< c^1< c^2<\ldots< c^{k-p}< b^1< b^2<\ldots< b^p.$$ Let $y_i$ be the farthest facility in $Y$ from $i$. Since we are on a line, it must be either the leftmost or rightmost facility of $Y$. If it is the rightmost facility of $Y$, then $i$ must be to the left of it, and $cost_i(Y)=d(i,y_i)\leq d(i,b^p)$, which means that $cost_i(Y)\leq cost_i(O)$ since $b^p\in O$, giving us the desired result. Thus we can assume that $y_i=a^1$, the leftmost facility of $Y$. 

        By triangle inequality, we have that $cost_i(Y) = d(i,a^1) \leq d(i,m(i))+d(m(i),a^1).$ By Lemma \ref{lem:2contiguous}, we also know that $d(m(i),a^1)\leq d(m(i),b^1)$, since $m(i)\in L$. Thus 
        \[
        \begin{aligned}
                    cost_i(Y) &\leq d(i,m(i)) + d(m(i),b^1)\\
                    &\leq d(i,m(i)) + d(i,m(i)) + d(i,b^1)\\
                    &\leq 2d(i,m(i)) + cost_i(O)
        \end{aligned}
        \]
        The last inequality above is because $b^1\in O$ and $i$ is a max-type agent. To bound $d(i,m(i))$, let $o^1$ be the leftmost facility of $O$ (so it is $b^1$ if $Y$ and $O$ have no intersection, and $c^1$ otherwise). We consider two subcases: {\bf Case 4a:} $o^1\leq m(i)$. Since $m(i)\in L$ and $i\in R$, this means that $o^1\leq m(i)\leq i$, so $d(m(i),i)\leq d(i,o^1)\leq cost_i(O)$, since $o^1\in O$ and $i$ is a max-type agent. {\bf Case 4b:} $m(i)< o^1$. For this case, we know that {\em all} facilities of $O$ are to the right of $m(i)$, since $o^1$ is the leftmost such facility. Therefore, $d(m(i),o^1)$  is the minimum distance from $m(i)$ to any facility in $O$, and so $d(m(i),o^1)\leq d(m(i),O)$. Therefore, combined with the fact that $i$ is a max-type agent, $d(m(i),i)\leq d(m(i),o^1)+d(o^1,i) \leq d(m(i),O) + d(i,o^1) \leq cost_{m(i)}(O) + cost_i(O)$. 

        In both cases 4a and 4b, we have that $d(m(i),i)\leq cost_{m(i)}(O) + cost_i(O)$. Therefore, combining this with the reasoning above and recall that $cost_i(Y)\leq 2d(i, m(i))+cost_i(O)$, we obtain that
        \[
        \begin{aligned}
            cost_i(Y)+cost_{m(i)}(Y) &= d(i,y_i) + d(m(i),Y)\\
            &\leq 2d(i,m(i))+ cost_i(O) + d(m(i),Y)\\
            &\leq 2cost_{m(i)}(O) + 2 cost_i(O) + cost_i(O) + d(m(i),Y)\\
            &\leq 3cost_i(O) + 3cost_{m(i)}(O)
        \end{aligned}
        \]
        The last inequality above is due to the fact that for $m(i)\in L$, we know that $cost_{m(i)}(Y)\leq cost_{m(i)}(O)$, as usual.

        Now that we have discussed all four cases, we can conclude that $cost_i(Y)+cost_{m(i)}(Y)\leq 3[cost_i(O)+cost_{m(i)}(O)]$, as desired.
    \end{proof}

    Most of the difficulty of this proof is in establishing the above lemma for mixed-type agents. With this lemma, we can now show the desired approximation bound. Let $L'$ be the agents $i$ in $L$ which do not have any $m(j)=i$, i.e., $L' = L\setminus m(R)$. Then, we have that
\[
\begin{aligned}
    \cmix(Y) &= \sum_{i\in \C}cost_i(Y)\\
    &=\sum_{i\in R}\left[cost_i(Y)+cost_{m(i)}(Y)\right] + \sum_{i\in L'}cost_i(Y)\\
    &\leq \sum_{i\in R}3\left[cost_i(O)+cost_{m(i)}(O)\right] + \sum_{i\in L'}cost_i(O)\\
    &\leq 3\sum_{i\in\C} cost_i(O) = 3\cdot\cmix(O)
\end{aligned}
\]
The first inequality above is due to Lemmas \ref{lem:charging} and \ref{lem:3contiguous}. This completes the proof of our approximation bound.
\end{proof}

\section{Unrestricted Locations} \label{sec:unres}

In this section, we will discuss the special case where agents are in a 1D line metric, and facilities can be placed anywhere on this line. In other words, $\F$ is an infinite continuous set of all possible locations. This setting has been thoroughly studied for the case when $k=1$, and the results are usually much stronger than for the restricted location setting where $\F$ is finite. This is also true for $k>1$, as we discuss below.

We begin by establishing some nice properties for this setting. Note that here we assume that multiple facilities can be placed in the same exact location, i.e., on top of each other. If that is not allowed, then our results hold in the limit with facilities being placed infinitessimaly close to one another. 

\begin{theorem}
    On a line with unrestricted facility locations, there must exist an optimal solution for $\cmix$ such that all chosen facility locations are identical.\label{thm:unres1}
\end{theorem}

\begin{proof}
    The argument is exactly the same as that of Theorem \ref{thm:optnextto}, showing that if there is a gap between two facility locations in a solution, then we can create a new solution with a smaller gap without making this solution worse. Thus there always exists an optimal solution $O$ such that $O = \{o,o\cdots, o\}$ for some location $o$.
\end{proof}

Moreover, for this nice setting of unrestricted facility locations, placing all facilities on top of (or as close as possible to) the median agent actually results in the optimum solution for any mix of agent types.

\begin{theorem}\label{thm:typeline}
    On a line with unrestricted facility locations, placing all facilities on top of the median agent gives an optimal solution for {\cmix}, {\css}, and {\csm}.
\end{theorem}

\begin{proof} Due to Theorem \ref{thm:unres1}, we know that there always exists an optimal solution $O = \{o,o\cdots, o\}$ with all facilities at the same location. Note that for such a solution, since $O$ contains only one location, the cost for each agent is the same regardless of their type, i.e., $cost_i(O)=d(i,O)=d(i,o)$ for both max and sum agents. Thus to find the optimum solution, we simply need to choose the location $o$ which minimizes $\sum_{i\in \C}cost_i(O) = \sum_{i\in C}d(i,o)$. Choosing $o$ to be the location of the median agent minimizes this sum, and so $O$ is an optimal solution for {\cmix}, for any mix of agent types (and thus is also optimal for {\css} and {\csm}). 
\end{proof}



With the above properties of optimal solutions, we can now consider strategyproof mechanisms. For unknown agent types (but known locations), Theorem \ref{thm:typeline} implies that we can simply put all facilities at the median agent, and obtain an optimal solution for $\cmix$. This is a strategyproof mechanism, so we can find the exact optimal solution for the case when $k>1$, as long as the metric space is a line and facility locations are unrestricted. For unknown agent locations, we can similarly just put all facilities on top of the median agent. This is the standard median mechanism used for the case when $k=1$, it is strategyproof, and due to the above theorems it results in the optimum solution for $\cmix$. Therefore, for unrestricted facility locations on a line, the strategyproof median mechanism always finds the optimum solution, just as it does for the classic setting of $k=1$ and no different agent types.



\section{Conclusion}
In this paper, we analyzed the quality of solutions formed by deterministic strategyproof mechanisms for our setting. This is the first such analysis for facility location in which we must select $k$ facilities in order to minimize the social cost, with agents being either max or sum types of individual cost functions. For the case when the types are private but agent locations are public knowledge, our results are very general, applying to arbitrary metric spaces, arbitrary values of $k$, and an arbitrary mix of sum and max agents. These results show that, especially for smaller $k$, we can easily form solutions which are only a small constant away from optimum, and can do better if we know the fraction of the total agents which are max-type. For the case when agent locations are private, we extend existing results about the median mechanism to hold for arbitrary $k$, which requires a new analysis due to their being a mix of agent types instead of all agents being the same type. Our results are significant generalizations and additions to existing knowledge.

Our results show the importance of collecting information about agent types vs about their locations, and show that it is possible to produce good outcomes even without such information. Nevertheless, many open questions remain. The most conspicuous open question is to improve the gap between our upper and lower bounds for the case when agent types are private. As mentioned in the paper, is seems significantly more difficult to form lower bounds for this case as compared to the ``private agent locations'' case (see also \cite{gai2024two}, which has an upper bound of $2n+1$ when types are private, but a lower bound of only $7/6$), so new insights and constructions would be needed to form such lower bounds. A more general research direction would be to consider different types of agents, and form mechanisms with good approximations for more than just a mix of max and sum agents. For example, the authors believe that similar results should be possible when agents are any of the $k$ different $\ell$-centrum types, i.e., $cost_i(A)$ is the sum of the largest $\ell$ distances to facilities in $A$. It would be interesting to quantify the classes of agent types for which there always exist strategyproof mechanisms with good approximation ratios, even if the agent types are not known to the mechanism.

\subsection*{Acknowledgments}
This work was partially supported by National Science Foundation award CCF-2006286.

\bibliography{refs}


\end{document}